\documentclass[12pt,notitlepage]{iopart}

\usepackage{iopams}
\usepackage{amsthm}

\expandafter\let\csname equation*\endcsname\relax

\expandafter\let\csname endequation*\endcsname\relax

\usepackage{amsmath}

\usepackage{latexsym}
\usepackage{amsfonts}
\usepackage{amssymb}
\usepackage{cite}

\usepackage{bbm,dsfont}
\usepackage{dsfont}
\usepackage{graphicx}
\usepackage{hyperref}
\usepackage{hyperref}
\usepackage{verbatim}
 
\newtheorem{proposition}{Proposition}
\newtheorem{proposition?}{Proposition?}
\newtheorem{theorem}{Theorem}
\newtheorem{lemma}{Lemma}
\newtheorem{corollary}{Corollary}

\newtheorem{example}{Example}

\newcommand{\hi}{\mathcal{H}} %Hilbert space H
\newcommand{\his}{\mathcal{H}_{\mathcal{S}}}
\newcommand{\hir}{\mathcal{H}_{\mathcal{R}}}

\newcommand{\hik}{\mathcal{K}} %Hilbert space K
\newcommand{\hs}{\mathcal{H}_{\mathcal{S}}}

\newcommand{\Y}{\yen}

\newcommand{\lh}{\mathcal{L(H)}} %bounded linear operators
\newcommand{\lhs}{\mathcal{L}(\mathcal{H}_{\mathcal{S}})} %bounded linear operators on system Hilbert space

\newcommand{\lhr}{\mathcal{L}(\hir)} %bounded linear operators on apparatus Hilbert space
\newcommand{\lk}{\mathcal{L(K)}} %bounded linear operators on K

\newcommand{\ip}[2]{\left\langle\,#1\,{|}\,#2\,\right\rangle} %inner product
\newcommand{\ket}[1]{|#1\rangle} %ket
\newcommand{\bra}[1]{\langle#1|} %bra
\newcommand{\no}[1]{\left\|#1\right\|} %norm
 
\renewcommand{\tr}[1]{\textrm{tr}\left[#1\right]} %trace
\newcommand{\id}{\mathbbm{1}} %identity operator
\newcommand{\Esf}{\mathsf{E}}

\newcommand{\Sy}{\mathcal{S}}

\newcommand{\R}{\mathcal{R}}
\newcommand{\A}{\mathsf{A}}%generic observable
\newcommand{\E}{\mathsf{E}}%generic observable
\newcommand{\F}{\mathsf{F}}%generic observable
\newcommand{\I}{\mathcal{I}}

\begin{document}

\title[Approximating relational observables by absolute quantities]{Approximating relational observables by absolute quantities: a quantum accuracy-size trade-off}

\author{Takayuki Miyadera}
\address{Department of Nuclear Engineering, Kyoto University, 
Kyoto daigaku-katsura, Nishikyo-ku, Kyoto, Japan 615-8540}
\ead{miyadera@nucleng.kyoto-u.ac.jp}
\author{Leon Loveridge}
\address{Department of Computer Science, University of Oxford, Wolfson Building, Parks Rd, Oxford, UK OX1 3QD}
\ead{leon.loveridge@cs.ox.ac.uk}
 \author{Paul Busch}
 \address{Department of Mathematics, University of York, Heslington, York, UK. YO10 5DD}
 \ead{paul.busch@york.ac.uk} 

  %\today
%\end{center}

%\date{\today}
\begin{abstract}
The notion that any physical quantity is defined and measured relative to a reference frame is traditionally not explicitly reflected in the theoretical description of physical experiments where, instead, the relevant observables are typically represented as ``absolute'' quantities. However, the emergence of the resource theory of quantum reference frames as a new branch of quantum information science in recent years has highlighted the need to identify the physical conditions under which a quantum system can serve as a good reference. Here we investigate
the conditions under which, in quantum theory, an account in terms of absolute quantities can provide a good approximation of relative quantities. We find that this requires the reference system to be large in a suitable sense.
\end{abstract}
%\pacs{03.65.Ta}
{\small Published in: J. Phys. A: Math. Theor. {\bf 49} 185301 (2016)}\\
{\small DOI: \href{http:dx.doi.org/10.1088/1751-8113/49/18/185301}{10.1088/1751-8113/49/18/185301}} 
\maketitle

\section{Introduction}

Symmetry plays a fundamental role in constructing and understanding physical theories \cite{Wigner,Weyl}. It also constrains the relationship between theoretical terms in a given formalism
and the world those terms are used to describe. In a theory with (say, a gauge)  symmetry, the quantities which may be measured (the {\em observables}) must be invariant with respect to the relevant symmetry transformations \cite{brs,lbm1}. Yet, it appears that in certain circumstances, measurements of symmetric systems can be well described in terms of non-invariant quantities. Typically, the invariant quantities are relative observables of a system plus reference, whereas the non-invariant quantities provide a simplified, ``absolute'' description of the situation in terms of the variables of the system alone. 

Problems of this kind appear naturally in investigations of the universality of quantum mechanics and the interface between quantum and
classical systems. For instance, in \cite{Page, Milburn, Brunetti, Giovannetti} 
it is argued that the notion of time appears from stationary observables 
in a composite system of the system and the clock. Another example, closer to the type of situation we will consider, arises in the theory of superconductivity: though the ``absolute" phase in 
a superconductor does not represent an observable
quantity, the Josephson effect shows that the relative phase between two systems can be measured,
since the relative phase operator is gauge (phase-shift) invariant. If one of these systems
is large in a suitable sense, the fluctuation of the phase can be neglected and the phase can be treated effectively
as a classical quantity. In this case the (statistics of a) relative phase measurement can be well approximated
by the statistics of a non-invariant phase observable of a subsystem. 
The large system is then viewed as a (classical) reference system. 

The possibility of such an approximation is considered in various forms in the literature and is accompanied by various interpretations. In \cite{brs}, a review on quantum reference frames and their use as resources for information processing,
 it is shown that any self-adjoint operator $A$ can be ``relativised" to give an (invariant) observable $\Y(A)$ acting in the Hilbert space of a larger system composed of the original system and a {\em reference system} (henceforth we use the term {\em system-plus-reference}). Then $A$ may approximate $\Y(A)$ well if the reference system has a localisable phase-like quantity \cite{lbm1}. $A$ is viewed as an ``approximator" for the invariant observable $\Y(A)$ in the high reference phase localisation limit. The large reference requirement also appears in the Wigner-Araki-Yanase theorem \cite{wig1, buschtrans, ay1, yan1, mi1, lb1, lb2} as a large spread in the apparatus part of a conserved quantity, and as an asymmetric resource state in \cite{ms1}, where the measurement of a non-invariant operator is viewed as being
simulated by the measurement of the invariant observable of system-plus-reference. The symmetry constraint and reference-system size also arise in work concerning superselection rules \cite{www,as,brs,lbm1}.

However, the problem of recasting relative quantities as ``absolute" ones, along with an operational
analysis of the ensuing accuracy (of the ``absolute" quantity as an approximator for the relative one) versus reference size, has hardly been addressed. The purpose, therefore, of this paper, is to prove that the operational ``distance" between an arbitrary effect (operator) of some quantum system and the restriction of a symmetry-invariant effect of system-plus-reference depends explicitly on the spread of the symmetry generator in the state of the reference system. For good approximation, the reference frame system must be large.  Our approach is quantitative and has a clear operational meaning, applies to unsharp as well as sharp observables, and holds for general continuous symmetry groups. We find, moreover, that the trade-off relations established are rooted in the uncertainty relations for certain incompatible quantities, and that therefore the size-versus-accuracy trade-offs are genuinely quantum constraints.

After reviewing standard background material in Section \ref{sec:mp}, we briefly introduce in Section \ref{sec:res} the notion of ``restriction" 
of operators on the system-plus-reference to  operators on the system, which arises by fixing a state of the reference.
Section \ref{sec:ppo} contains the relativisation map introduced in \cite{lbm1}, accompanied by a discussion of the role of reference system localisation. Section \ref{sec:ga} brings our main result: a trade-off between precision and size. This is given as a quantitative relation for the distance between arbitrary and restricted invariant effects; we show that the distance may be made small only when the reference state is delocalised with respect to the symmetry generator or, where applicable, localised with respect to its conjugate quantity. The origin of the trade-off as arising from quantum incompatibility is discussed, and we conclude in Section \ref{sec:con} with a summary and remarks on future work.

\section{Mathematical preliminaries} \label{sec:mp}
We briefly present some mathematical concepts and notation used throughout the paper.
 Any Hilbert space $\mathcal{H}$ we consider will be finite-dimensional and we write $\dim \hi$ for the dimension of $\mathcal{H}$. The standard inner product is denoted 
 $\ip{\cdot}{\cdot}: \mathcal{H} \times \hi \to \mathbb{C}$ (taken to be linear in the second argument). The algebra of bounded  linear operators $A: \hi \to \hi$ will be denoted by $\lh$, with the \emph{positive} operators ($A\ge 0$)  given as those $A \in \lh$ for which $\langle\varphi|A|\varphi\rangle\ge0$ for all $\varphi\in\hi$ (if $A - B \geq 0$ we write $A \geq B$ or $B \leq A$).  \emph{States} in $\hi$, (equivalently, on $\lh$) are given by the positive operators $\rho$ in $\lh$ for which the trace $\tr{\rho}=1$ (equivalently, positive linear functionals $\omega: \lh \to \mathbb{C}$ with $\omega(\id) = 1$), the convex set of which is written $\mathcal{S(H)}$. The extreme elements of $\mathcal{S(H)}$---the pure states---satisfy $\rho ^2 = \rho$, and are given by rank one projections, or simply as unit vectors in $\hi$. 
 
 An \emph{observable} $\Esf$ is a \emph{normalised, positive operator valued measure} ({\sc pom}) $\Esf:\mathcal{F} \to \lh$ on some measurable space $(\Omega,\mathcal{F})$ (so that $\Esf(X)\ge 0$, $\Esf(\Omega)=\id$, and $\Esf$ satisfies the additivity property of measures).
Any operator $\Esf(X)$ in the range of $\Esf$ is an \emph{effect}, that is, a positive operator $A$ for which $0 \leq A \leq \id$;
we write $\mathcal{E(H)}$ for the (convex) set of all effects in $\hi$. States are then
generalised probability measures on $\mathcal{E}(\hi)$. 
An effect $A$ satisfies $V(A):=\no{A-A ^2}=0$ (where $\no{\cdot}$ denotes  the operator norm) if and only 
if $A$ is a projection, otherwise $V(A)$ is a measure of the \emph{noise}, or \emph{unsharpness}, of the effect $A$ 
\cite{buschunsharp, miuncertainty, polterovich}.
The difference between two effects $A$ and $B$ in a given state $\omega$ can be
quantified by $D_{\omega}(A, B):=|\omega(A)-\omega(B)|$. A (state-independent) distance function $D$ is then obtained by taking the supremum over all states:
\begin{eqnarray*}%\label{eq:D}
D(A,B):=\sup_{\omega}D_{\omega}(A, B) = \sup_{\omega} |\omega(A)-\omega(B)|\, =\Vert A -B\Vert.
\end{eqnarray*}
Since $\omega(A)$ and $\omega(B)$ are probabilities, $D$ has a clear operational meaning \cite{miuncertainty}.

Let $G$ be a group. A system is said to {\em possess $G$ as a symmetry (group)} 
if there is a projective (anti-)unitary representation $\{U(g)\}_{g\in G}$ of $G$ on $\hi$. 
That is, the $U(g)$ are (anti-)unitary operators 
satisfying $U(g) U(g') = \mu(g, g') U(g g')$ with  
coefficients $|\mu(g,g')|=1$ 
The projective representation induces 
a $*$-automorphism $\alpha_g : \lh \to \lh$ by 
$\alpha_g(A)= U(g) A U(g)^*$, which satisfies 
$\alpha_g \circ \alpha_{g'} = \alpha_{g g'}$.  
In this paper, we consider a finite-dimensional 
connected Lie group $G$, in which case each $U(g)$ is automatically unitary.
We denote by $\frak{g}$ the Lie algebra of $G$. 
There exists a neighbourhood $O$ of $e$ (the unit of $G$)
such that every element $g \in O$ is uniquely expressed
 as $g = exp(n )$ with $n \in \frak{g}$. 
Once we fix $n \in \frak{g}$, we can define a line 
 $g(s) = exp(s n)$ with a variable real parameter $s$. 
This satisfies $g(s)g(t) = g(s+t)$ for small enough $|s|$ and $|t|$. 
$U(g)$ can be set so as to satisfy 
$U(g(s))U(g(t))=U(g(s+t))$.
Then Stone's theorem on unitary representations of 
one-parameter Abelian groups applies, to obtain 
$U(e^{s n}) = e^{i s N}$ with some self-adjoint operator $N$ (a generator) 
on the Hilbert space. We say an effect $A$ is $G$-{\em invariant} if 
 $U(g)A U(g)^* = A$ for 
all $g\in G$. This naturally implies $[{N}, A]=0$ 
for each generator ${N}$. Similarly, a state $\rho$ is ($G$-)invariant if
$U(g)^*\rho U(g) = \rho$ for all $g \in G$, equivalently if $[{N},\rho]=0$, or if $\omega = \omega_{\rho}: \lh \to \mathbb{C}$ satisfies $\omega_{\rho}(A) = \omega_{\rho}(U(g)AU(g)^*)$ for all $A \in \lh$ and $g \in G$.

Given Hilbert spaces $\hi$ and $\mathcal{K}$  (as ``output" and ``input space", respectively), 
we consider channels in the Heisenberg picture.
A linear map $\Lambda: \lh \to \lk$ is called 
a {\em channel} if (i) $\Lambda(\id)=\id$ and (ii) $\Lambda \otimes id :
\lh \otimes \mathbb{C}^d \to \lk \otimes \mathbb{C}^d$ is 
positive for all $d$ (that is, $\Lambda$ is {\em completely positive}). In the Schr\"odinger picture 
a state $\rho \in \mathcal{S}(\hik)$ is mapped to 
$\Lambda_* (\rho) \in \mathcal{S}(\hi)$, which is defined 
by $\mbox{tr}\bigl[\rho \Lambda (A)\bigr] = \mbox{tr}\bigl[\Lambda_* (\rho) A\bigr]$
for all $A \in \lh$.
If a group $G$ has  (projective) representations $U_1,U_2$ on 
$\hi$ and $\hik$, respectively, a channel $\Lambda$ is called  
($G$-){\em covariant} if it satisfies 
$\Lambda \bigl(U_1(g)AU_1(g)^*\bigr)=U_2(g)\Lambda (A) U_2(g)^*$ for all $A \in \lh$ and $g \in G$.
$G$-covariant channels map invariant effects to invariant effects.

\section{Restriction map}\label{sec:res}

Let us consider a system $\Sy$ and a reference $\R$ with Hilbert spaces $\his$ and $\hir$, respectively. 
 The Hilbert space of the 
total system (system-plus-reference) 
is $\hi:=\his\otimes \hir$.
As a preliminary observation we note that given a self-adjoint $A \in \lh \equiv \mathcal{L}(\his \otimes \hir)$, and by holding $\phi \in \hir$ (assumed to be a unit vector) fixed, the expression
$\ip{\cdot \otimes
 \phi}{A| \cdot \otimes \phi}$ defines a bounded (real) quadratic form $\Phi : \his \times \his \to \mathbb{C}$
and thus a unique self-adjoint operator $A_{\phi} \in \lhs$. This then induces a
 {\em restriction map} of self-adjoint operators 
in $\lh$ to those in $\lhs$, relative to some unit vector in $\hir$, which we observe is unital and completely positive, and which gives rise to a restriction $\mathcal{E}(\hi) \to \mathcal{E}(\his)$ of effects (and of {\sc pom}s). The state $\phi$ plays the role of a ``reference state" through which one can ``view" effects of $\his$. %\footnote{cf. the ``probability reproducibility" condition in the quantum theory of measurement, where $\phi$ is viewed as a fixed apparatus preparation and for a given self-adjoint operator (``pointer observable") $Z$ (with spectral resolution $P^Z$) of the apparatus and unitary evolution $U$, the condition $\ip{U(\varphi \otimes \phi)}{\id \otimes P^Z(X) U \varphi \otimes \phi} = \ip{\varphi}{\Esf(X) \varphi}$, to hold for all outcome sets $X$ and unit $\varphi \in \his$, defines a unique measured {\sc pom} $\Esf$.} 
%One is then able to compare arbitrary effects $A \in \mathcal{E}(\his)$ to those 
%of the form $\Esf_{\phi}$, where $\Esf$ is constrained to be invariant under symmetry.
%For notational convenience we will drop the subscript and leave the $\phi$-dependence %implicit.
%Before considering the role of symmetry, and 
More generally, fixing a state $\omega_{\R} \in \mathcal{S}(\hir)$ we define the restriction map 
$\Gamma_{\omega_{\R}}: \lh \to \mathcal{L}(\hs)$ via 
\begin{equation*}
\Gamma_{\omega_{\R}}(A \otimes B)=A \omega_{\R}(B),
\end{equation*}
extended by linearity to the whole of $\lh$.
For any $\omega_{\R}$, $\Gamma_{\omega_{\R}}$ is unital and completely positive. Of course, $\Gamma_{\omega_{\R}}(A \otimes \id) = A$. 
Conversely, any channel 
$\Gamma: \lh \to \mathcal{L}(\hs)$ satisfying 
$\Gamma( A \otimes \id) = A$ has 
a unique state $\omega_{\R}$ on $\hir$ 
satisfying $\Gamma(A \otimes B) = A \omega_{\R}(B)$.\footnote{To see this, we introduce 
a sesquilinear map
$\lh \times \lh \to \mathcal{L}(\his)$ 
by $\langle X, Y\rangle := \Gamma(X^* Y) -\Gamma(X^*)\Gamma(Y)$. 
It satisfies a Cauchy-Schwarz-type inequalilty: 
$\langle X, Y\rangle \langle Y, X\rangle \leq \Vert \langle Y, Y\rangle\Vert 
\langle X, X\rangle$ as shown in Section \ref{sec:ga} (Lemma \ref{lem:cs}).
Since $\langle A \otimes \id, A \otimes \id\rangle =0$, 
we have 
$\langle A\otimes \id, \id \otimes B
\rangle =0$. It implies 
$\Gamma(A \otimes B) = A \Gamma(\id \otimes B)$ 
and $\Gamma(A \otimes B) = \Gamma(\id \otimes B) A$ 
for all $A$ and $B$. 
Thus $\Gamma(\id \otimes B)$ is proportional to $\id$. 
Since $\Gamma$ is a channel, there exists 
a state $\omega_{\R}$ satisfying $\Gamma(\id \otimes B) 
= \omega_{\R}(B)$.} Physically, $\Gamma_{\omega_{\R}}$ may be understood
as providing a reduced description in terms of $\Sy$ alone, contingent upon the state
of the reference being $\omega_{\R}$.
The Schr\"odinger picture is 
helpful to see the physical meaning of the restriction map. 
It maps a state $\omega$ on the system to 
$\omega\circ \Gamma
= \omega \otimes \omega_{\R}$. 
Thus $\Gamma$ describes the usual ``tracing out" 
operation for product states. %Interestingly, as has already been shown  in a model in \cite{as}, proximity to product states with one factor representing a highly localised reference system can be achieved with dynamics that strictly obey the given symmetry requirements.
Note that the restriction map is not $G$-covariant in general for 
non-invariant states $\omega_\R$.

We now state our problem. Let $G$ (a connected Lie group) be a symmetry 
of the system and the reference; that is, there exist 
projective unitary representations $U_{\Sy}$ and $U_{\R}$ on 
$\his$ and $\hir$ respectively, with the tensor product representation written
$U_{\Sy}\otimes U_{\R}$ (i.e., $(U_{\Sy}\otimes U_{\R})(g) = U_{\Sy}(g)\otimes U_{\R}(g)$ for each $g \in G$).
For each (possibly non-invariant) $A \in \mathcal{E}(\his)$, 
how close can $\Gamma(E)$ be to $A$ by choosing a
globally invariant $E \in \mathcal{E}(\hi)$ (that is, an $E$ for which 
$(U_{\Sy}(g) \otimes U_{\R}(g)) E (U_{\Sy}(g) \otimes U_{\R}(g))^* 
= E$ for all $g\in G$) and restriction map $\Gamma$? 

We begin by considering the special case of invariant $E$ obtained by relativisation, before
addressing the general case.

\section{Relativisation map}\label{sec:ppo}
In \cite{lbm1} it is shown that {\sc pom}s, and therefore self-adjoint operators and effects, 
can be ``relativised" to give associated invariant (with respect to some chosen groups) {\sc pom}s, self-adjoint operators and effects in a new Hilbert space. For instance, position relativises to relative position, angle to relative angle, and phase (as a {\sc pom}) to relative phase. We therefore view the relativisation procedure as giving rise to relative observables and effects even when defined on arbitrary {\sc poms} (or effects). 
Here we study relativisation for the case (of representations of) $G=U(1)$ (otherwise called the \emph{phase group}) to give general quantitative 
tradeoff relations between (in)accuracy and reference system localisation. 

Specifically, let $N_{\Sy}$ and $N_{\R}$ denote number operators acting in $\his$ and $\hir$. 
Thus, for example, $N_{\Sy}$ is a self-adjoint operator on $\his$ with orthonormal basis of eigenvectors
$\{ \ket{n} \}_{n \in I} \subset \his$, $I = \{0,...,N_0 -1 \}$, so that $N_{\Sy} = \sum_n n\ket{n}\bra{n} \equiv \sum_n n P_n$.
We define  $N:= N_{\Sy}+N_{\R}$ on $\hi=\his\otimes\hir$, and denote the  associated  unitary 
groups by $U_{\Sy}(\theta) = e^{i N_{\Sy} \theta}$, $U_{\R}(\theta)=e^{i N_{\R} \theta}$ and 
$U(\theta)=e^{i N \theta} = U_{\Sy}(\theta)\otimes U_{\R}(\theta)$; $\theta \in [-\pi, \pi)$.

Stipulating that only invariant, or \emph{relative} effects are observable \cite{lbm1}, we may then only measure those
effects $E \in \mathcal{E}(\hi)$ for which 
$U(\theta)E U(\theta)^* = E$ or, equivalently, $[E,N]=0$. We may construct relative effects in $\mathcal{E}(\hi) =\mathcal{E}( \his \otimes \hir)$ out of arbitrary effects in $\mathcal{E}(\his)$; first we introduce
the linear mapping (an adapted and generalised form of that given in \cite{brs}) $\Y: \lhs \to \lh$ (see below for the full definition):
\begin{equation*}\label{eq:yen}
\Y(A) = \int_{-\pi}^{\pi} U_{\Sy}(\theta) A U_{\Sy}(\theta)^* \otimes \F (d \theta).
\end{equation*}
Here $\F$ is a covariant phase {\sc pom} on the reference (i.e., acting in $\hir$), 
defined as any  {\sc pom} on (the Borel algebra $\mathcal{B}\bigl([-\pi, \pi)\bigr)$ of) $[-\pi,\pi)$
for which
\begin{equation*} \label{eq:cph}
U_{\R}(\theta)\F(X)U_{\R}(\theta)^* = \F (X \dotplus \theta),
\end{equation*}
where $\dotplus$ denotes addition modulo $2 \pi$.

Here we give the definition of $\Y(A)$ and investigate its 
properties. 
For a covariant phase {\sc pom} $\F$   
there exists \cite{Davies} a uniquely determined positive operator $T \geq 0$
satisfying 
for all 
measurable subsets $X$, 
\begin{eqnarray*}
\F(X) = \frac{1}{2\pi} \int_X d\theta U_{\R}(\theta) T U_{\R}(\theta)^*, 
\end{eqnarray*}
with $\frac{1}{2\pi}
\int^{\pi}_{-\pi} d\theta\  U_{\R}(\theta) T U_{\R}(\theta)^* = \id$.
That is, for each state $\omega_{\R}$, there exists a positive 
smooth function 
$f_{\omega_{\R}}(\theta):= \frac{1}{2\pi} \omega_{\R}( 
U_{\R}(\theta) T U_{\R}(\theta)^*)$ satisfying 
$\omega_{\R}(\F(X)) = 
\int_X d\theta \ f_{\omega_{\R}}(\theta)$. 
We denote the 
absolutely continuous measure $\omega_\R\circ \F$ by $\mu_{\omega_\R}^\F$, so that 
$\mu_{\omega_{\R}}^{\F}(d\theta) = f_{\omega_{\R}}(\theta) d\theta$.   
$\Y(A)$ is defined by 
\begin{eqnarray*}
\Y(A) = \frac{1}{2\pi} 
\int^{\pi}_{-\pi}d\theta U_{\Sy}(\theta) A U_{\Sy}(\theta)^*
\otimes U_{\R}(\theta) T U_{\R}(\theta)^*. 
\end{eqnarray*}
It is easy to see that $\Y(A)$ is $G$-invariant
and becomes an effect for any effect $A$.
Again, the Schr\"odinger picture is helpful to see 
the physical meaning of $\Y(A)$. 
In the Schr\"odinger picture, the state must be invariant 
while arbitrary effects are allowed. 
In this picture 
the simplest (thus product) relative effect has the form $A \otimes T$. 
Switching to the Heisenberg picture, we obtain 
$\Y(A)$. 

$\Gamma_{\omega_{\R}}(\Y(A))$ is given by
\begin{eqnarray*}
\Gamma_{\omega_{\R}}(\Y(A))= 
\int \mu^{\F}_{\omega_{\R}}(d\theta) \ U_{\Sy}(\theta) A U_{\Sy}(\theta)^*.
\end{eqnarray*}

Quantifying the distance between an effect $A \in \mathcal{E}(\his)$ and a restriction of a 
relativised version, {\em viz.} $\Gamma_{\omega_{\R}} (\Y (A))$, amounts to estimating the following quantity:
\begin{equation*}
D\left(A,\Gamma_{\omega_{\R}}\bigl(\Y(A)\bigr)\right) = \biggl \Vert \int_{-\pi}^\pi
 \mu_{\omega_{\R}}^{\F}(d \theta) 
 \bigl(e^{i\theta N_{\Sy}}A e^{-i \theta N_{\Sy}} -A\bigr)\biggr \Vert.
\end{equation*}

For $0 \leq \epsilon < 1$, the {\em overall width} at confidence level $1-\epsilon$ of the measure $\mu_{\omega_{\R}}^{\F}$ 
is defined by  
\begin{equation}
%W_{\epsilon}(\mu_{\omega_{\R}}^{\F}):= \inf \left\{w \big | \exists  \theta:\,\mu_{\omega_{\R}}^{\F}\bigl( [-\tfrac{w}{2}+ \theta, 
%\tfrac{w}{2}+\theta]\bigr)\geq 1-\epsilon\right\}. 
W_{\epsilon}(\mu_{\omega_{\R}}^{\F}):= \inf \left\{w \big | \exists  \theta:\,\mu_{\omega_{\R}}^{\F}\bigl( \mathcal{I}(\theta,w)\bigr)\geq 1-\epsilon\right\},
\end{equation}
where $\mathcal{I}(\theta,w)$ denotes the closed
interval of width $w\le2\pi$ centred at $\theta$ in $[-\pi,\pi)$, understood as a circle.
It can be shown (see, {\em e.g.}, \cite[Chapter 12]{QMMT}) that 
the above infimum is actually a minimum. Hence,
due to the absolute continuity of $\F$, there is a $\theta_0$ such that 
$\mu_{\omega_\R}^{\F}\left(
\I(\theta_0, W_{\epsilon}(\mu_{\omega_{\R}}^{\F})
\right)=1-\epsilon$. 
We call $\I(\theta_0, W_{\epsilon}(\mu_{\omega_{\R}}^{\F}))$ 
an $\epsilon$-support of the measure $\mu_{\omega_{\R}}^{\F}$ (noting that it need not be unique).
In addition, we define another quantity $W^0_{\epsilon}(\mu_{\omega_\R}^\F)$---the overall width around $0$---by 
\begin{eqnarray*}
W^0_{\epsilon}(\mu_{\omega_\R}^\F)
:=  \inf \left\{w \big |\ \mu_{\omega_{\R}}^{\F}\bigl( \mathcal{I}(0,w)\bigr)\geq 1-\epsilon\right\}. 
\end{eqnarray*}
Of course, $W^0_{\epsilon}(\mu_{\omega_\R}^\F) 
\geq W_{\epsilon}(\mu_{\omega_\R}^\F)$ holds. 
Note that this $\inf$ can also be replaced by $\min$. 
That is, we have 
\begin{equation}\label{eq:eps-width}
\mu_{\omega_{\R}}^{\F}\left(\mathcal{I}(0, W^0_{\epsilon}(\mu_{\omega_\R}^\F)) \right) =1-\epsilon.
\end{equation}

\begin{proposition}
Let $\Y$ be a relativisation map and $\Gamma_{\omega_{\R}}$ a restriction map.
For an arbitrary effect $A$ and $0\leq \epsilon <1$, it holds that
\begin{eqnarray*}
D\bigl(A, \Gamma_{\omega_{\R}}(\Y (A))\bigr)
\leq \bigl\Vert [N_{\Sy}, A]\bigr\Vert  
\left( \tfrac12{W^0_{\epsilon}(\mu_{\omega_\R}^{\F})}(1-\epsilon) 
+ \pi \epsilon \right). 
\end{eqnarray*}
\end{proposition}
\begin{proof}

Let $\theta \in [-\pi, \pi)$. 
\begin{eqnarray*}
\frac{d}{d\theta}
e^{i N_\Sy \theta} A e^{-iN_\Sy \theta} 
= e^{i N_\Sy\theta} i [N_\Sy, A] e^{-i N_\Sy\theta} 
\end{eqnarray*}
is used to obtain
\begin{eqnarray*}
e^{i N_\Sy\theta} A e^{-iN_\Sy \theta} - A 
= \int^{\theta}_0 d\theta' e^{i N_\Sy \theta'} 
i [N_\Sy, A] e^{- i N_\Sy \theta'}. 
\end{eqnarray*}
Taking the norm yields
\begin{eqnarray*}
\bigl\Vert e^{i N_\Sy \theta} A e^{- i N_\Sy \theta}-A \bigr\Vert 
\leq |\theta| \bigl\Vert [N_\Sy, A] \bigr\Vert. 
\end{eqnarray*}

Thus,
\begin{align*}
D\bigl(A, \Gamma_{\omega_{\R}}(\Y (A))\bigr)
&\leq
\bigl\Vert [N_{\Sy}, A]\bigr\Vert 
\int _{-\pi}^\pi \mu_{\omega_\R}^{\F}(d\theta) |\theta|
\\
&=
\bigl\Vert [N_{\Sy}, A]\bigr\Vert
\left(
\int^{W^{0}_{\epsilon}(\mu^{\F}_{\omega_{\R}})/2}_{-W^0_{\epsilon}(\mu^{\F}_{\omega_{\R}})/2}
\mu^{\F}_{\omega_\R}(d\theta)|\theta| \right.\\
&\left. \hspace{2cm}
+ 
\int^{\pi}_{W^0_{\epsilon}(\mu^{\F}_{\omega_{\R}})/2}
\mu^{\F}_{\omega_\R}(d\theta)|\theta| 
+\int^{-W^0_{\epsilon}(\mu^{\F}_{\omega_{\R}})/2}_{-\pi}
\mu^{\F}_{\omega_\R}(d\theta) |\theta|
\right)
\\
&\leq \bigl\Vert [N_{\Sy}, A]\bigr\Vert 
\left( \tfrac12{W^0_{\epsilon}(\mu_{\omega_\R}^{\F})}
(1-\epsilon) + \pi \epsilon \right), 
%\qquad \red{\text{\bf ?}}
\end{align*}
where we used Eq.~\eqref{eq:eps-width}. 
\end{proof}
Therefore, we may conclude that for $\omega_\R$ with strong localisation 
around $\theta =0$, the discrepancy between 
$A$ and $\Gamma_{\omega_{\R}} (\Y (A))$ is small. 
%%%%%

An $\omega_{\R}$ which has small overall width but large 
overall width around $0$ does not give good approximation of $A$ by $\Y(A)$.
However, we can construct a deformed quantity $\Y_{\theta_0}(A)$ 
for $\theta_0 \in [-\pi, \pi)$ as
\begin{eqnarray*}
\Y_{\theta_0}(A) 
= (\id \otimes U_{\R}(\theta_0)) \Y(A) (\id \otimes U_{\R}(\theta_0)^*),  
\end{eqnarray*}
which is $G$-invariant. 
Recall that for $0\leq \epsilon <1$, each $\omega_{\R}$ has 
$\theta_0$ such that 
$\mu^{\F}_{\omega_\R}(\I(\theta_0, W_{\epsilon}(\mu^{\F}_{\omega_{\R}})))= 1-\epsilon$. 
For such a $\theta_0$, we have 
\begin{eqnarray*}
D\bigl(A, \Gamma_{\omega_{\R}}(\Y_{\theta_0} (A))\bigr)
\leq \bigl\Vert [N_{\Sy}, A]\bigr\Vert  
\left( \tfrac12{W_{\epsilon}(\mu_{\omega_\R}^{\F})}(1-\epsilon) 
+ \pi \epsilon \right).
\end{eqnarray*}
\par
In the following we show that the high localisation 
of $\omega_{\R}$ is necessary to have $A$ well approximated 
 by $\Gamma_{\omega_{\R}}(\Y (A))$. 
In this illustration we assume that $N_{\Sy}$ has eigenstates $\ket{0}$ and $\ket{1}$ with respective eigenvalues 
$0$ and $1$.
We consider a particular effect $A$ 
defined by $A= \frac{1}{2}\bigl(
|1\rangle\langle 1| + |0\rangle 
\langle 0| + 
|0\rangle \langle 1 | +|1\rangle \langle 0|\bigr)$. 
The effect $A$ does not commute with $N_{\Sy}$. 
We show that the magnitude of difference between $A$ and 
$\Gamma_{\omega_{\R}}(\Y (A))$ is related to the localisation property indeed. 
%%%%%
%%%%%%%%%%
\begin{proposition}\label{prop2}
For $A= \frac{1}{2}\bigl(|1\rangle \langle 1| 
+ |0\rangle \langle 0| +|0\rangle \langle 1| + |1\rangle \langle 0|\bigr)$, it holds that 
\begin{eqnarray*}
D\bigl(A, \Gamma_{\omega_{\R}}( \Y(A))\bigr)
\geq \frac{\epsilon}{2}\Bigl( 1- \cos \left (\tfrac12{W^0_{\epsilon}(\mu_{\omega_\R}^{\F})}\right) \Bigr)  
%(1-\epsilon). 
\end{eqnarray*} 
\end{proposition}
\iffalse\begin{proof}
As $e^{i N_{\Sy} \theta} A e^{-i N_{\Sy} \theta } 
-A = \frac{1}{2}\bigl[|0\rangle \langle 1 | (e^{-i \theta} -1 ) 
+ ( e^{i \theta } -1) |1\rangle \langle 0 |\bigr] $ holds, 
we have 
\begin{eqnarray*}
D\bigl( A, \Gamma_{\omega_{\R}}(\Y (A))\bigr)
= \tfrac{1}{2}\bigl\Vert |0\rangle \langle 1 | 
\overline{c} + c |1\rangle \langle 0|\bigr\Vert =\tfrac12{|c|}, 
\end{eqnarray*}
where $c:= \int \mu_{\omega_\R}(d\theta) ( 1- e^{i \theta})$.
Using Chebyshev's inequality, $|c|$ is estimated for $1\geq t >0$ as
\begin{align*}
|c|\geq \mbox{Re}\  c 
&= \int \mu_{\omega_\R}^{\F}(d\theta) (1- \cos \theta)\\
&\geq t \cdot \mu_{\omega_\R}^{\F} \bigl(\{\theta | 
1- \cos \theta \geq t\}\bigr)
\\
&= t \left[ 1- \mu_{\omega_\R}^{\F}\bigl(\{\theta | 1- \cos \theta < t\}\bigr)\right]. 
\end{align*}
We put $t = 1- \cos\left( \frac12{W^0_{\epsilon}(\mu_{\omega_\R}^\F)}\right)$ to 
obtain (using again equality \eqref{eq:eps-width})
\begin{eqnarray*}
t \left[ 1- \mu_{\omega_\R}^{\F}\bigl(\{\theta | 1- \cos \theta < t\}\bigr)\right]
= \left[ 1- \cos\left( 
\tfrac12{W^0_{\epsilon}(\mu_{\omega_\R}^{\F}) }\right)\right]
\epsilon.
\end{eqnarray*}
 \end{proof}
%%%%%%%%%%%%
\fi
%\blue{Comment: replaced old proof by new one.}
\begin{proof}
As $e^{i N_{\Sy} \theta} \A e^{-i N_{\Sy} \theta } 
-\A = \frac{1}{2}\bigl[|0\rangle \langle 1 | (e^{-i \theta} -1 ) 
+ ( e^{i \theta } -1) |1\rangle \langle 0 |\bigr] $ holds, 
we have 
\begin{eqnarray*}
D\bigl( A, \Gamma_{\omega_{\R}}(\Y (A))\bigr)
= \tfrac{1}{2}\bigl\Vert |0\rangle \langle 1 | 
\overline{c} + c |1\rangle \langle 0|\bigr\Vert =\tfrac12{|c|}, 
\end{eqnarray*}
where $c:= \int \mu^{\mathsf{F}}_{\omega_\R}(d\theta) ( 1- e^{i \theta})$.
This $|c|$ is estimated for $t\in \mathbb{R}$ as
\begin{align*}
|c|\geq \mbox{Re}\  c 
&= \int \mu_{\omega_\R}^{\F}(d\theta) (1- \cos \theta)\\
&= \int_{1-\cos \theta >t} \mu_{\omega_\R}^{\F}(d\theta)
(1-\cos\theta)
+ \int_{1-\cos \theta \leq t} \mu_{\omega_\R}^{\F}(d\theta)
(1-\cos \theta) \\
&\geq t \cdot \mu_{\omega_\R}^{\F} \bigl(\{\theta | 
1- \cos \theta > t\}\bigr)
\\
&= t \left[ 1- \mu_{\omega_\R}^{\F}\bigl(\{\theta | 1- \cos \theta \leq t\}\bigr)\right]. 
\end{align*}
We put $t = 1- \cos\left( \frac12{W^0_{\epsilon}(\mu^{\F}_{\omega_\R})}\right)$ to 
obtain (using again \eqref{eq:eps-width})
\begin{eqnarray*}
t \left[ 1- \mu_{\omega_\R}^{\F}\bigl(\{\theta | 1- \cos \theta \leq t\}\bigr)\right]
= \left[ 1- \cos\left( 
\tfrac12{W^0_{\epsilon}(\mu_{\omega_\R}^{\F}) }\right)\right]
\epsilon.
\end{eqnarray*}
 \end{proof}

%end \blue
\noindent
Thus we can see that as $W^0_{\epsilon}(\mu_{\omega_\R}^{\F})$ becomes large
so does the discrepancy $D\bigl(A, \Gamma_{\omega_{\R}}( \Y(A))\bigr)$. 
%%%%%%%%%%%%%%%%%%

%\begin{remark}
General effects can be treated in essentially the same manner. 
Suppose $A$ does not commute with $N_{\Sy}= \sum_n n P_n$, where each $P_n$ is a projection.
$A$ can be decomposed into two parts, as
$A= \sum_n P_n A P_n + 
\sum_{n \neq m} P_n A P_m$,  and therefore
$e^{i N_{\Sy} \theta} A e^{-i N_{\Sy} \theta } 
-A= \sum_{n > m} ( (e^{i (n-m)} -1) P_n A P_m
+(e^{-i(n-m)} -1) P_m AP_n ) 
$ holds.  
Putting $c_{nm}:= \int \mu^{\F}_{\omega_{\R}}(d\theta) (1-e^{i(n-m)\theta})$, 
we have 
\begin{eqnarray*}
D\bigl( A, \Gamma_{\omega_{\R}}(\Y (A))\bigr)
=  \left\Vert \sum_{n>m} 
c_{nm} P_n A P_m + \overline{c}_{nm} P_m A P_n
\right\Vert.  
\end{eqnarray*}
As there exists a pair $n>m$ with $P_n AP_m \neq 0$ 
due to the noncommutativity between $A$ and $N_{\Sy}$, 
we can choose normalized vectors $|n\rangle $ and $|m\rangle$ 
with $P_n |n\rangle = |n\rangle $ and $P_m|m\rangle = |m\rangle$  
so that $|\phi\rangle:= \frac{1}{\sqrt{2}}(|n\rangle + |m\rangle)$ 
satisfies $\ip{\phi}{\left(c_{nm} P_n A P_m + \overline{c}_{nm} P_m AP_n \right)| \phi}
= \frac{1}{2} \left( c_{nm} \ip{n}{A|m}  
+ \overline{c}_{nm} \ip{m}{A|n} \right) = |c_{nm}|| \ip{n}{A|m}|\neq 0$.
Thus we have a bound, 
\begin{eqnarray*}
D\bigl( A, \Gamma_{\omega_{\R}}(\Y (A))\bigr)
\geq |c_{nm}| |\ip{n}{A|m} |. 
\end{eqnarray*}
$|c_{nm}|$ is bounded as
\begin{eqnarray*}
|c_{nm}| \geq \mbox{Re}\ c_{nm} 
\geq \int^{\frac{\pi}{n-m}}_{- \frac{\pi}{n-m}} \mu^{\F}_{\omega_{\R}}(d\theta) 
(1-\cos (n-m) \theta)
\\
\geq t
\cdot \mu^{\F}_{\omega_{\R}}\left(\{\theta | 
1- \cos (n-m) \theta > t\} \cap [-\frac{\pi}{n-m}, \frac{\pi}{n-m}]\right).
\end{eqnarray*}
For sufficiently large $\epsilon$ satisfying $W_{\epsilon}^0(\mu^{\F}_{\omega_{\R}}) 
\leq \frac{2\pi}{n-m}$, we put $t= 
1- \cos\left( \frac{n-m}{2}W^0_{\epsilon}(\mu^{\F}_{\omega_\R})\right)$ 
to obtain
\begin{eqnarray*}
t
\cdot \mu^{\F}_{\omega_{\R}}\left(\{\theta | 
1- \cos (n-m) \theta > t\} \cap [-\frac{\pi}{n-m}, \frac{\pi}{n-m}]\right)\\
=
\left[1- \cos\left( \frac{n-m}{2}W^0_{\epsilon}(\mu^{\F}_{\omega_\R})\right)
\right] \epsilon.
\end{eqnarray*}
Thus $D\bigl( A, \Gamma_{\omega_{\R}}(\Y (A))\bigr)$ can be 
bounded by a similar inequality in which 
$W^0_{\epsilon}(\mu^{\F}_{\omega_{\R}})$ plays a role. 
In the following we treat only the special effect 
$A= 
\frac{1}{2}(|0\rangle \langle 0| + |1\rangle \langle 1| 
+ |0\rangle \langle 1| +|1\rangle \langle 0|)$ for simplicity. 

%\end{remark}

%%%%%%%%%
In \cite{lbm1} it is shown that a large reference frame allows for good phase localisation around 0. Now we show that to attain good localisation in the sense of small overall width, a large Hilbert space dimension for the reference is needed.
\begin{lemma}\label{lem:owe}
The overall width $W^0_{\epsilon}(\mu_{\omega_\R}^{\F})$ of $\mu_{\omega_\R}^{\F}$ around $0$
satisfies, for any $\epsilon\geq 0$, 
\begin{eqnarray*}
1-\epsilon \leq \frac{\dim \hir}{2\pi} W^0_{\epsilon}(\mu_{\omega_\R}^{\F}).
\end{eqnarray*}
\end{lemma}
\begin{proof}
For any $X\subset [-\pi, \pi)$ and for any state $\omega_\R$, 
we have $\omega_\R(\F(X)) \leq \mbox{tr}[\F(X)]$, where 
$\mbox{tr}$ denotes the trace in $\hir$.
Now for any eigenstates $|n\rangle $ of $N_{\R}$, 
 %it holds that  $\langle n | \F(X) |n\rangle = \langle n| \F(X + \theta) |n\rangle$, therefore 
we have $\ip{n}{\F(X)| n} = |X|/2\pi$ and so 
$\mbox{tr}[\F(X)] = \dim\hir |X|/2\pi$.  We put 
$X = [-W^0_{\epsilon}( \mu^{\F}_{\omega_\R})/2, W^0_{\epsilon}(\mu^{\F}_{\omega_\R})/2]$
to obtain the result (using once more the equality \eqref{eq:eps-width}).
\end{proof}
Note that this estimate is not tight. 
For instance, for $\epsilon=0$ we obtain 
$W^0_0(\mu_{\omega_\R}^{\F}) \geq \frac{2 \pi}{\dim \hir}$. 
However, one can show that for any state $W^0_0(\mu_{\omega_\R}^{\F})= 2\pi$. 
That is, no state can be strictly localised \cite{numphasecomplem}. 
In fact, if we were to assume that a state $\omega_\R$ could 
be strictly localised, there would exist an open set $X \subset [-\pi, \pi)$ 
for which $\omega_\R(\F(X))=0$. 
Defining $f(\theta):= \omega_\R\bigl( e^{i N_{\R} \theta} \F(X) e^{-i N_{\R} \theta}\bigr)$ and noting that the Hilbert space is finite dimensional, 
this function can be analytically continued to the whole 
complex plane. This analytic function, by assumption, satisfies
 $f(\theta)=0$ on some open interval of $\mathbb{R}$. 
This means that $f(\theta) =0$ on the whole plane, which contradicts
$\omega_\R\bigl(\F([-\pi, \pi))\bigr) =1$. 
%%%%%%%
\begin{proposition}
Let $A =\frac{1}{2}\bigl( |1\rangle \langle 1| + 
|0\rangle \langle 0| + |0\rangle \langle 1 | + |1\rangle \langle 0|\bigr)$. 
For any $\epsilon \geq 0$, 
\begin{eqnarray*}
D\bigl( A, \Gamma_{\omega_{\R}}(\Y (A))\bigr)
\geq \frac{1}{2} \left( 1 - \cos \left ( \pi\frac{1- \epsilon}{\dim \hir}\right)
\right) \epsilon.
\end{eqnarray*} 
\end{proposition}
\begin{proof}
We use Lemma \ref{lem:owe} and Proposition \ref{prop2} to obtain 
$ \cos \left( {W^0_{\epsilon}(\mu_{\omega_\R}^{\F})}/{2}
\right) \leq \cos \left ( \pi(1- \epsilon) /{\dim \hir}\right)$. 
\end{proof}
In particular for $\epsilon = \frac{1}{2}$, we obtain 
\begin{eqnarray*}
D\bigl(A, \Gamma_{\omega_{\R}}(\Y (A))\bigr) \geq \frac{1}{4}
 \left( 1 - \cos \left ( \frac{ \pi}{2 \dim \hir }\right)
\right), 
\end{eqnarray*}
which is a non-trivial bound for $\dim \hir <\infty $.

Let us estimate the overall width around $0$ in terms of 
the standard deviation of the number operator 
$\Delta_{\omega_\R} N_{\R}:= 
(\omega_{\R}(N_{\R}^2) - \omega_{\R}(N_{\R})^2)^{1/2}$. 
We obtain the following bound for $W^0_{\epsilon}(\mu^{\F}_{\omega_\R})$.
\begin{lemma}\label{lemma2}
%Assume $0 \leq \epsilon \leq \frac{1}{2}$, 
Assume 
$W^0_{\epsilon}(\mu^{\F}_{\omega_\R})\leq \pi$ and 
%Assume 
$\Delta_{\omega_\R} N_{\R} \cdot 
W^0_{\epsilon}(\mu^{\F}_{\omega_\R})\leq \frac{\pi}{2}$. 
Then $W^0_{\epsilon}(\mu^{\F}_{\omega_\R})$ is bounded as
\begin{eqnarray*}
\cos\left( \Delta_{\omega_{\R}} N_{\R}\cdot  
W^0_{\epsilon}(\mu^{\F}_{\omega_\R})
\right) \leq \sqrt{\epsilon(1-\epsilon)}+ \sqrt{\epsilon}. 
\end{eqnarray*}
\end{lemma}
\begin{proof}
For $\epsilon$ satisfying $\sqrt{\epsilon (1-\epsilon)} + \sqrt{\epsilon}
\geq 1$ the claim is trivial and therefore we assume otherwise. 
In the following $W^0_{\epsilon}(\mu_{\omega_{\R}}^{\F})$ 
is denoted by $W^0_{\epsilon}$ for notational simplicity.  
As stated above, the definition of the overall width around $0$ 
implies  
$\omega_{\R}
(\F\left( \I(0, W^0_{\epsilon})\right))
=\mbox{tr}[\rho_{\R}\F\left( \I(0, W^0_{\epsilon})\right)]=1-\epsilon$, where $\rho_{\R}$ is a density operator representing 
$\omega_{\R}$, i.e., $\omega_{\R}(A)= \mbox{tr}[\rho_{\R}A]$ for all $A \in \lhr$.
Now we define shifted states (as density operators)
of $\rho_{\R}$ by 
$\rho_{\R, \theta}:= e^{-i N_{\R}\theta} \rho_{\R} e^{i N_{\R}\theta}$ 
for $-\pi\le\theta\le\pi$. 
Let us consider the fidelity \cite{Fuchs,Nielsen} between 
$\rho_{\R}$ and $\rho_{\R, - W^0_{\epsilon}}$. 

The fidelity between two states $\rho_0$ and $\rho_1$ is defined by 
$F(\rho_0, \rho_1):= \mbox{tr}\left[\bigl(\rho_0^{1/2}\rho_1 \rho_0^{1/2}\bigr)^{1/2}\right]$, 
which takes values in $[0,1]$. 
It quantifies the closeness of two states such that  
$F(\rho_0, \rho_1)=1$ holds if and only if $\rho_0=\rho_1$. 
We note that the fidelity has another representation 
\cite{Jozsa}, 
\[
F(\rho_0, \rho_1)=\inf_{\E: \mbox{{\sc pom}}} 
\sum_x \mbox{tr}[\rho_0 \E(\{x\})]^{1/2} \mbox{tr}[\rho_1 \E(\{x\})]^{1/2},
\] 
where $\inf$ is taken over all discrete {\sc pom}s.

%With $\I(0, W^0_{\epsilon})^c:= [-\pi, \pi) 
%\setminus \I(0, W^0_{\epsilon})
%= \mbox{Int}\bigl(\I(\pi, 2\pi- W^0_{\epsilon})\bigr)$,\footnote{For a set 
%$S$, the symbol $\mbox{Int}(S)$  denotes the {\em interior} of $S$, 
%which comprises all elements of $S$ that are not boundary points; 
%note that $\mbox{Int}(S)$ is an open set.}
%we have $\mbox{tr}[\rho_{\R} \F(\I(0, W^0_{\epsilon})^c)] = \epsilon$, 
%$\mbox{tr}[\rho_{\R, - W^0_{\epsilon}}
%\F(\I(0, W^0_{\epsilon})^c]
%=\mbox{tr}[\rho_{\R}\F(\I(\pi - W^0_{\epsilon}, 2\pi - W^0_{\epsilon}))]
%\geq \mbox{tr}[\rho_{\R} \F(\I(0, W^0_{\epsilon})]=1-\epsilon$, 
%where we used $\I(\pi - W^0_{\epsilon}, 2 \pi-W^0_{\epsilon}) 
%\supset \I(0, W^0_{\epsilon})$, 
%and thus $\mbox{tr}[\rho_{\R, -W^0_{\epsilon}} \F(\I(0, W^0_{\epsilon}))] \leq \epsilon$. 
With  $\I(0, W^0_{\epsilon})^c:= [-\pi, \pi) \setminus \I(0, W^0_{\epsilon})$,
we have $\mbox{tr}\bigl[\rho_{\R} {\F(\I(0, W^0_{\epsilon})^c})\bigr] = \epsilon$. Further, noting that 
the ``shifted" (modulo $2\pi$) set ${\I(0, W^0_{\epsilon})^c}-W^0_{\epsilon}={\I(-W^0_\epsilon, W^0_{\epsilon})^c}$
contains $\I(0, W^0_{\epsilon})\setminus \{ -W^{0}_{\epsilon}/2\}$ (since $W^0_\epsilon\le\pi$), 
we also obtain
$\mbox{tr}\bigl[\rho_{\R, - W^0_{\epsilon}} \F(\I(0, W^0_{\epsilon})^c)\bigr]
=\mbox{tr}\bigl[\rho_{\R}\F(\I( - W^0_{\epsilon}, W^0_{\epsilon})^c)\bigr]
\geq \mbox{tr}\bigl[\rho_{\R} \F(\I(0, W^0_{\epsilon}))\bigr]=1-\epsilon$, 
%\red{
%T: 
%Rigorously speaking, a point $\theta = -W^{0}_{\epsilon}/2 \in 
%\I(0, W^0_{\epsilon})$ is not contained in 
%$\I(-W^0_{\epsilon}, W^0_{\epsilon})^c$ as the latter is an open set. 
%But this point does not give any contribution on the measure.
%}
and thus $\mbox{tr}[\rho_{\R, -W^0_{\epsilon}} \F(\I(0, W^0_{\epsilon}))]
\leq \epsilon$.    
Therefore we have 
\begin{eqnarray*}
%F(\rho_{\R}, \rho_{\R}(W^0_{\epsilon}(\mu^{\F}_{\omega_\R}))
%\leq \mbox{tr}[\rho_{\R}\F(\I(0, \frac{W^0_{\epsilon}
%(\mu^{\F}_{\omega_\R})}{2}))]^{1/2}
%\mbox{tr}[\rho_{\R}\F(\I(0, \frac{W^0_{\epsilon}
%(\mu^{\F}_{\omega_\R})}{2}))]^{1/2}
F(\rho_{\R}, \rho_{\R, -W^0_{\epsilon}})
&&
\leq \mbox{tr}[\rho_{\R}\F(\I(0, W^{0}_{\epsilon}))]^{1/2}
\mbox{tr}[\rho_{\R, -W^0_{\epsilon}} \F(\I(0, W^{0}_{\epsilon}))]^{1/2}
\\
&&
+ 
\mbox{tr}[\rho_{\R}\F(\I(0, W^{0}_{\epsilon})^c)]^{1/2}
\mbox{tr}[\rho_{\R, -W^0_{\epsilon}} \F(\I(0, W^{0}_{\epsilon})^c)]^{1/2}
\\
&&
\leq \sqrt{\epsilon(1-\epsilon)} + \sqrt{\epsilon}. 
\end{eqnarray*}
The Mandelstam-Tamm uncertainty relation \cite{Mandelstam, BuschTime} 
claims that for an arbitrary normalized vector $|\phi\rangle  \in \hir$ 
and $\theta$ with $\Delta_{\phi  } N_{\R} \cdot |\theta | \leq \pi/2$, 
$|\ip{ \phi}{e^{-i N_{\R} \theta} |\phi} | \geq \cos \left( 
\Delta_{\phi} N_{\R} \cdot \theta \right)$, 
where $\Delta_{\phi} N_{\R}
= \left( \ip{ \phi}{ N_{\R}^2| \phi} 
- \ip{\phi}{ N_{\R}|\phi} ^2 \right)^{1/2}$. 
%\red{Should be an equality as variance has already been defined?} 
Since the fidelity between general states $\rho_0$ and $\rho_1 $ is 
also given as $F(\rho_0, \rho_1) = \sup |\langle \phi_0 | \phi_1\rangle |$ 
where $\sup$ is taken over all the possible 
purifications of $\rho_0$ and $\rho_1$, 
the Mandelstam-Tamm inequality 
is applied to purified states to give \cite{MiTime}, 
\begin{eqnarray*}
F(\rho_{\R}, \rho_{\R, -\theta})
\geq \cos
\left(
\Delta_{\omega_{\R}}N_{\R} \cdot \theta 
\right). 
\end{eqnarray*}
Thus we obtain 
\begin{eqnarray*}
\cos\left( \Delta_{\omega_{\R}} N_{\R} \cdot
W^0_{\epsilon}\right) \leq \sqrt{\epsilon(1-\epsilon)}+ \sqrt{\epsilon}. 
\end{eqnarray*}
\end{proof}

\begin{theorem}
Let $A$ be an effect defined by $A=\frac{1}{2}(|0\rangle \langle 0 | 
+ |1\rangle \langle 1| + |0\rangle \langle 1 | + |1\rangle \langle 0|)$.  
For $\omega_{\R}$ satisfying $\Delta_{\omega_{\R}} N_{\R} < \frac{1}{6}$,
\begin{eqnarray*}
D(A, \Gamma_{\omega_{\R}}(\Y(A))) > \frac{1}{32}.
\end{eqnarray*}
For $\omega_{\R}$ satisfying 
$\Delta_{\omega_\R} N_{\R} \geq \frac{1}{6}$, 
it holds 
\begin{eqnarray*}
D(A, \Gamma_{\omega_{\R}}(\Y(A))) 
\geq \frac{1}{32}\left(
1- \cos\left( 
\frac{\pi}{12\Delta_{\omega_{\R}}N_{\R}}\right)
\right). 
\end{eqnarray*}
\end{theorem}
\begin{proof}
We put $\epsilon = \frac{1}{16}$. 
There are three possibilities: 
(i) $W^0_{\frac{1}{16}}(\mu^{\F}_{\omega_{\R}})\leq \pi$ 
and $\Delta_{\omega_{\R}} N_{\R} \cdot W^0_{\frac{1}{16}}
(\mu^{\F}_{\omega_{\R}})>\frac{\pi}{2}$, 
(ii) $W^0_{\frac{1}{16}}(\mu^{\F}_{\omega_{\R}})\leq \pi$ 
and $\Delta_{\omega_{\R}} N_{\R} \cdot W^0_{\frac{1}{16}}
(\mu^{\F}_{\omega_{\R}}) \leq \frac{\pi}{2}$,
(iii) $W^0_{\frac{1}{16}}(\mu^{\F}_{\omega_{\R}}) > \pi$. 
(i) implies $\Delta_{\omega_{\R}}N_{\R} > \frac{1}{2}$. 
In this case, Proposition \ref{prop2} gives 
\begin{eqnarray}
D(A, \Gamma_{\omega_{\R}}(\Y(A))) 
\geq \frac{1}{32}\left(
1- \cos\left( 
\frac{\pi}{4\Delta_{\omega_{\R}}N_{\R}}\right)
\right). 
\label{case1}
\end{eqnarray}
Under condition (ii) lemma \ref{lemma2} applies. 
For $\epsilon = \frac{1}{16}$, 
$\sqrt{\epsilon (1-\epsilon)} + \sqrt{\epsilon} 
\leq 2 \sqrt{\epsilon} = \frac{1}{2}$ holds and 
thus 
\begin{eqnarray}
\Delta_{\omega_{\R}} N_{\R} \cdot 
W^0_{\frac{1}{16}} ( \mu^{\F}_{\omega_\R}) \geq \frac{\pi}{6}
\label{case2}
\end{eqnarray}
follows. It implies $\Delta_{\omega_{\R}} N_{\R} \geq \frac{1}{6}$ and 
\begin{eqnarray*}
D(A, \Gamma_{\omega_{\R}}(\Y(A))) 
\geq \frac{1}{32}\left( 1- \cos\left( \frac{\pi}{12\Delta_{\omega_{\R}}N_{\R}}\right)\right). 
\end{eqnarray*}
Combining (i) and (ii), we observe 
that $\Delta_{\omega_{\R}} N_{\R} 
<\frac{1}{6}$ implies $W^0_{\frac{1}{16}}(\mu^{\F}_{\omega_\R}) >\pi$, thus 
(iii). 
In this case, Proposition \ref{prop2} is applied to show that
\begin{eqnarray}
D(A, \Gamma_{\omega_{\R}}(\Y(A)) > \frac{1}{32}. 
\label{case3}
\end{eqnarray}
This completes the proof. 
%As (\ref{case1}) implies (\ref{case2}), and (\ref{case3}) implies 
%(\ref{case2}), it ends the proof. \red{Is this the correct order of implications?}
\end{proof}

To summarise, we showed that 
in this relativisation model the strong localisation (that is, small overall width) 
of the reference state is necessary and sufficient to obtain good approximation. 
Moreover, to achieve the strong localisation we need a large reference, in the sense of large Hilbert space dimension.
This last finding can be naturally interpreted by the uncertainty relation for  
joint localisability \cite{BuschOperational, miJMP}.

There are various reasons for going beyond these results. $\Y$ relies on covariant phases, and does not yield all invariant observables in $\hi$, immediately pointing to the need to understand the general case beyond the particular construction. We now present a model-independent trade-off between the operational measure of distance $D(\Gamma (E),A)$ of an arbitrary effect $A$ from the restriction of an invariant effect $E$, and the size of the reference system. 

%%%%%%
\section{General argument}\label{sec:ga}

In this section we provide a general operational trade-off relation. We consider an arbitrary effect $A \in \mathcal{E}(\his)$ and a relative (invariant) effect $E \in \mathcal{E}(\his \otimes \hir)=\mathcal{E}(\hi)$
and try to bound $D(A, \Gamma_{\omega_{\R}}(E))$. Once more there is a 
symmetry (Lie) group $G$
acting (via projective unitary representations) both in $\mathcal{L}(\his)$ and $\mathcal{L}(\hir)$, and we recall that for each element $n$ in the Lie algebra $\frak{g}$ of $G$, there exist corresponding
self-adjoint operators $N_\Sy$ and $N_\R$ 
such that for sufficiently small $|s|$, 
$\alpha^\Sy_{e^{n s}}(A) = e^{i {N_\Sy} s} A e^{-i {N_\Sy} s}$
and $\alpha^\R_{e^{n s}}(A) = e^{i N_\R s} A e^{- i N_\R s}$. 

For notational simplicity we write $\Gamma:{\lh} \to \lhs$ for a channel of the form $\Gamma = \Gamma_{\omega_{\R}}$ for some $\omega_{\R}$.  %We also denote the standard deviation of $N_\R$ in the state $\omega_\R$ by $\Delta_{\omega_\R}N_\R:=\left(\omega_{\R}(N_{\R}^2)-\omega_{\R}(N_{\R})^2\right)^{1/2}$. 
The following is our main result.

\begin{theorem}\label{th:tradeoff} Recall that $V(A)=\Vert A-A^2\Vert$, $D(A,B) = \no{A-B}$,
$N_{\Sy}$ and $N_{\R}$ are number operators on $\his$ and $\hir$ respectively, and $\omega_{\R} \in \mathcal{S}(\hir)$. Then the following inequality holds:
\begin{align*}
\bigl\Vert [A, N_\Sy]\bigr\Vert 
%&
\leq %\blue{2?} 
2 D\bigl(\Gamma(E), A\bigr) \Vert N_\Sy\Vert
%\\
%&
+2 \left(\omega_{\R}(N_{\R}^2)-\omega_{\R}(N_{\R})^2\right)^{1/2}
%\\
%&
\,\bigl(2D(\Gamma(E),A)+V(A)\bigr)^{1/2}.
\end{align*}
\end{theorem}
This inequality shows that good proximity between a symmetric, relative observable and a non-symmetric absolute quantity requires a large spread in the reference system's number operator.

Before we prove the theorem, we recall some relevant results concerning channels $\Lambda:\mathcal{L}(\hi) \to \mathcal{L}(\hik)$. In the application of the proof, we will specify 
$\Lambda=\Gamma:\mathcal{L}(\his\otimes\hir)\to\mathcal{L}(\his)$.
 
%\blue{[(T): $\hi$ has already appeared. 
%To avoid possible confusions there will be two options. 
%(i) Use $\hik_1$ and $\hi_2$ instead of 
%$\hi$ and $\hik$. And indicate later when 
%we replace them and $\Lambda$ by $\Gamma$. 
%(ii) Treat $\Gamma$ itself instead of general $\Lambda$. 
%; I prefer the first one although the second one is sufficient.... ]}\\
%\blue{[(Pa): alternatively we can simply swap $\hi$ and $\hik$ in the coming paragraphs %before the proof of Thm 1 starts. This is also consistent with the notation used 
%initially in the def of a channel.]}
%and consider its (not necessarily minimal) Stinespring representation 
%$(\hik',V)$. 
 
Define a sesquilinear mapping $\langle \cdot,\cdot\rangle:\mathcal{L}(\hi)
\times \mathcal{L}(\hi) \to\mathcal{L}(\hik)$ by
\begin{equation*}
%\langle a\otimes m, b\otimes n\rangle
%:=\Lambda((a\otimes m)^* (b\otimes n))
%-\Lambda ((a\otimes m)^*)\Lambda(b\otimes n). 
%\langle A,B\rangle:=\Lambda(A^* B)-\Lambda(A^*)\Lambda(B). 
%\langle A,B\rangle:=V^*A^*BV-V^*A^*VV^*BV. 
\langle A,B\rangle :=\Lambda(A^*B) -\Lambda(A^*)\Lambda(B). 
\end{equation*}
This mapping satisfies $\langle A,A\rangle\geq 0$ and
$\langle A,B\rangle^* =\langle B,A\rangle$ for all $A, B\in\mathcal{L}(\hi)$.
Thus this can be regarded as an ``operator-valued inner product''. 
It satisfies a Cauchy-Schwarz inequality proven by Janssens \cite{Ja06}, which we recapitulate  here as we expect it to be useful for future steps in this line of investigation.
%%%%%
\begin{lemma}\label{lem:cs}
Consider the map $\langle \cdot, \cdot  \rangle$ defined above. 
For any $A,B \in \lh$ , it satisfies 
\begin{equation*}
\langle A,B\rangle \langle B,A\rangle \leq 
\Vert \langle B,B\rangle \Vert \langle A,A\rangle.
\end{equation*}
\end{lemma}
%%%%%%
\begin{proof}
We consider a (not necessarily minimal) Stinespring representation 
$(\hi',V)$ of the channel $\Lambda$. 
That is, $\hi'$ is a Hilbert space and $V: \hik \to \hi \otimes \hi'$ 
is an isometry such that $\Lambda(A)= V^*(A\otimes \id_{\hi'}) V$ holds.
By introducing the operator 
$\xi:=\sqrt{\id-VV^*}$, 
we can represent the map as 
\begin{equation*}
\langle A,B \rangle=V^*(A^*\otimes \id_{\hi'})\xi^*\xi (B \otimes \id_{\hi'})V. 
\end{equation*} 
We then have 
\begin{equation*}
\langle A, B\rangle \langle B,A\rangle
=V^*(A^*\otimes \id_{\hi'}) \xi^*\xi (B\otimes \id_{\hi'})VV^* 
(B^*\otimes \id_{\hi'}) \xi^* \xi (A\otimes \id_{\hi'})V. 
\end{equation*}
As $\xi (B\otimes \id_{\hi'})VV^* (B^*\otimes \id_{\hi'})\xi^* \geq 0$, 
it follows that 
\begin{align*}
V^*(A^*\otimes \id_{\hi'}) \xi^*\xi (B\otimes \id_{\hi'})
&VV^* (B^*\otimes \id_{\hi'}) \xi^* \xi (A\otimes \id_{\hi'})V 
\\
&\leq \Vert \xi (B\otimes \id_{\hi'})VV^* (B^*\otimes \id_{\hi'})
 \xi^*\Vert V^*(A^*\otimes \id_{\hi'}) \xi^*\xi (A\otimes \id_{\hi'})V. 
\end{align*}
Using the $C^*$ property of the norm, we obtain 
\begin{align*}
\Vert \xi (B\otimes \id_{\hi'}) VV^* (B^*\otimes \id_{\hi'}) \xi^*\Vert 
=\Vert V^*(B^*\otimes \id_{\hi'}) \xi^*\xi (B\otimes \id_{\hi'})V \Vert
=\Vert \langle B, B\rangle\Vert \,,
\end{align*}
which proves the lemma.
\end{proof}

The following lemma is an immediate consequence of the previous result.
\begin{lemma}\label{normCS}
For any $A,B\in \lh$,
\begin{equation*}
\Vert \langle A,B\rangle \Vert^{2}
% \langle B,A\rangle 
\leq 
\Vert \langle B,B\rangle \Vert \Vert \langle A,A\rangle \Vert.
\end{equation*}
\end{lemma}
We now have the following result.
\begin{lemma}\label{th:uncertain}
%Let $\Phi_1\hh,\Phi_2\hh:\lk\to\lh$ be two compatible operations.  
If $A,B\in \lh$ satisfy $[A,B]=0$, then
\begin{multline}\label{eq:incomprehensible_inequality}
\Vert
 [\Lambda(A),\Lambda(B)]
\Vert
\leq
\Vert 
\Lambda(A^*A)-\Lambda(A)^*\Lambda(A)\Vert^{1/2}
\Vert
\Lambda(BB^*)-\Lambda(B)\Lambda(B)^*\Vert^{1/2}
\\
+
\Vert 
\Lambda(AA^*)-\Lambda(A)\Lambda(A)^*\Vert^{1/2}
\Vert
\Lambda(B^*B)-\Lambda(B)^*\Lambda(B)\Vert^{1/2}.
\end{multline}  
\end{lemma}

\begin{proof}
Appealing again to the sesquilinear mapping $\langle A, B\rangle$, we write
\begin{align*}
\langle A^\ast , B\rangle 
=\Lambda(AB)-\Lambda(A)\Lambda(B),\ 
\langle B^*, A\rangle 
=\Lambda(BA)-\Lambda(B)\Lambda(A).
\end{align*}
Thus for $A,B\in \lh$ satisfying $[A,B]=0$, 
it holds that 
\begin{align*}
[\Lambda(A),\Lambda(B)]=\langle B^*,A\rangle -\langle A^*,B\rangle.
\end{align*}
By Lemma \ref{normCS}, the result is proved.
\end{proof}

We are now ready to prove Theorem \ref{th:tradeoff}.
\begin{proof}
Define $\epsilon:=\Gamma(E)-A$; we wish to estimate $\| \epsilon \| = D(\Gamma(E),A)$. 
Now we have $[A,N_{\Sy}]=[\Gamma(E),N_{\Sy}]+[N_{\Sy} ,\epsilon]$ and thus 
\begin{eqnarray}
\Vert [A,N_{\Sy}]\Vert \leq \Vert [\Gamma(E),N_{\Sy}]\Vert +\Vert [N_{\Sy}, \epsilon]\Vert. 
\label{ANS}
\end{eqnarray}
We provide a bound for each term. 
The second term of the right-hand side is 
easily bounded as
$\Vert [N_{\Sy}, \epsilon] \Vert \leq 2 \Vert N_{\Sy} \Vert \Vert \epsilon \Vert$. 

By assumption we have $[E,N]=0$, and therefore we may apply Lemma \ref{th:uncertain} in order to bound the first term on the right hand side of (\ref{ANS}). Note that $\Gamma(N)=N_{\Sy}+\omega_{\R}(N_{\R})\id$. 
We thus obtain 
\begin{align}
\bigl\Vert [\Gamma(E),N_{\Sy}]\bigr\Vert 
&=\bigl\Vert [\Gamma(E),\Gamma(N)]\bigr\Vert \nonumber \\
&\leq
2 \bigl\Vert \Gamma(E^2)-\Gamma(E)^2\bigr\Vert^{1/2}
\bigl\Vert \Gamma(N^2)-\Gamma(N)^2\bigr\Vert^{1/2}
\nonumber \\
&=
2 \bigl\Vert \Gamma(E^2)-\Gamma(E)^2\bigr\Vert^{1/2}
\bigl(\omega_{\R}(N_{\R}^2)-\omega(N_{\R})^2\bigr)^{1/2}. 
\label{AA}
\end{align}
We now bound $\Vert \Gamma(E^2)-\Gamma(E)^2\Vert$. 
Using $\Gamma(E^2)-\Gamma(E)^2\geq 0$ (two-positivity) 
and $E^2 \leq E$, we obtain 
\begin{align*}
\Gamma(E^2)-\Gamma(E)^2 &\leq \Gamma(E)-\Gamma(E)^2
\\
&=(\epsilon+A)-(\epsilon+A)^2
\\
&=\epsilon-\epsilon^2 -\epsilon A- A \epsilon + A-A^2.
\end{align*}
Thus it holds that 
\begin{align*}
\Vert \Gamma(E^2)-\Gamma(E)^2\Vert
&\leq \Vert \epsilon-\epsilon^2 -\epsilon A- A \epsilon\Vert  + V(A)
\\
&=
\bigl\Vert [A,\epsilon] + (\id -2A-\epsilon)\epsilon\bigr\Vert +V(A)
\\
&\leq  
\bigl\Vert [A,\epsilon]\bigr\Vert +\bigl\Vert \bigl(\id-(A+\Gamma(E))\bigr)\epsilon\bigr\Vert +V(A) . 
\end{align*}
$0\leq A \leq \id$ gives $\bigl\Vert [A,\epsilon]\bigr\Vert \leq \Vert \epsilon\Vert$. 
In fact, $\bigl\Vert [A, \epsilon] \bigr\Vert 
= \sup_{\Vert \phi \Vert =1} 
| \ip{ \phi}{ i [A, \epsilon]| \phi} |$ 
can be bounded by Robertson's uncertainty relation as 
$|\ip{\phi}{ i [A, \epsilon]| \phi} | \leq 
2 \sqrt{  \ip{ \phi} {A^2| \phi}
- \ip{ \phi}{ A| \phi}^2} 
\Vert \epsilon \Vert$ with 
$\sqrt{ \ip{ \phi}{A^2| \phi} - 
\ip{ \phi}{A| \phi}^2} \leq \frac{1}{2}$,  
% 
% \blue{Should it be $\leq 2 \no{\epsilon}$? In which case below it would be
% $3||\epsilon||$} 
and 
$0\leq A + \Gamma(E)\leq 2\id$ gives 
$\bigl\Vert \id -(A+\Gamma(E))\bigr\Vert \leq 1$. Thus we obtain 
\begin{align}
\bigl\Vert \Gamma(E^2)-\Gamma(E)^2\bigr\Vert 
\leq 
2\Vert \epsilon \Vert +V(A). 
\label{GG}
\end{align}
Putting (\ref{GG}) in (\ref{AA}), we obtain 
\begin{align}
\Vert [\Gamma(E),N_{\Sy}]\Vert 
\leq 
2 (2\Vert \epsilon \Vert +V(A))^{1/2} 
(\omega_{\R}(N_{\R}^2)-\omega(N_{\R})^2)^{1/2}. 
\label{DF}
\end{align}
Thus, (\ref{ANS}) can be bounded as
\begin{align*}%\label{eq:compro}
\bigl\Vert [A,N_{\Sy}]\bigr\Vert
\leq 
2\Vert N_{\Sy}\Vert \Vert \epsilon\Vert 
+
2 \bigl(2\Vert \epsilon \Vert +V(A)\bigr)^{1/2} 
\bigl(\omega_{\R}(N_{\R}^2)-\omega(N_{\R})^2\bigr)^{1/2}. 
\end{align*}
This completes the proof.
\end{proof}
%%%%%
The following corollaries are immediate consequences. 
\begin{corollary} 
It holds that 
\begin{align*}
\bigl\Vert [A, N_{\Sy}]\bigr\Vert 
\leq 2 D(\Gamma(E), A) \Vert N_{\Sy}\Vert   %\quad \red{\text{ factor 2 missing?}}
%\\
%&
+2 \Vert N_{\R}\Vert \bigl(2D(\Gamma(E),A)+V(A)\bigr)^{1/2}.
\end{align*}
%where $V(A)=\Vert A-A^2\Vert$. 
\end{corollary}
%%%
\begin{corollary}\label{cor:sharp}
For a sharp effect (projection) $A$ it holds that 
%for each $X$, 
\begin{align*}
\bigl\Vert [A, N_{\Sy}]\bigr\Vert 
\leq 2 D(\Gamma(E), A) \Vert N_{\Sy}\Vert %    \quad \red{\text{ factor 2 missing?}}
%\\
%&
+2 \sqrt{2} \bigl(\omega_{\R}(N_{\R}^2) -\omega_{\R}( N_{\R})^2\bigr)^{1/2}\,
D(\Gamma(E),A)^{1/2}, 
\end{align*}
\end{corollary}
%%%%%%
Thus to attain a good measurement a large reference frame is required, 
since it is only then that $\omega_{\R}(N_{\R}^2) -\omega_{\R}( N_{\R})^2$ can be large.

\begin{example}\emph{ ($G= U(1)$).
In section \ref{sec:ppo} we studied the particular map $\Y$ for 
$G=U(1)$ and showed that a large reference frame is needed for 
$\Y(A)$ to well approximate the effect $A = \frac{1}{2}
(|0\rangle \langle 0 | + |1\rangle \langle 1| 
+ |0\rangle \langle 1| + |1\rangle \langle 0|)$. 
Below we estimate the lower bound of the distance 
$D(\Gamma (E), A)$ for an arbitrary 
invariant $E$. 
Since $A$ acts only on a subspace spanned by $\{|0\rangle, 
|1\rangle\}$, we introduce a projection operator $P= |0\rangle \langle 
0 | + |1\rangle \langle 1|$. 
We first note that for an arbitrary invariant $E$, we have
\begin{eqnarray*}
D\bigl(\Gamma ((P\otimes \id) E(P\otimes \id)), A\bigr)
\leq D\bigl(\Gamma (E), A\bigr).
\end{eqnarray*} 
In fact, since $\Gamma((P\otimes \id) E(P\otimes \id))
=P \Gamma(E) P$ holds,\footnote{
This follows from the footnote in Section \ref{sec:res}. 
}
$D\bigl(\Gamma((P\otimes \id) E (P\otimes \id)), A\bigr)
= \bigl\Vert P( \Gamma(E) - A) P \bigr\Vert \leq D(\Gamma(E), A)$ follows. 
Thus it suffices to consider only $E$ satisfying 
$(P\otimes \id) E (P\otimes \id) = E$. 
For such an $E$, $\epsilon:= \Gamma(E) -A$ satisfies 
$[N_{\Sy}, \epsilon] = [PN_{\Sy} P, \epsilon]$. 
Thus $\bigl\Vert [N_{\Sy}, \epsilon]\bigr\Vert $ in (\ref{ANS}) can be 
bounded as 
$\bigl\Vert [N_{\Sy}, \epsilon]\bigr\Vert \leq \Vert \epsilon\Vert$ 
(where positivity of $N_{\Sy}$ is used). 
Compared to the direct application of Corollary \ref{cor:sharp}, this gives the tighter bound
\begin{eqnarray*}
\frac{1}{2} \leq D(\Gamma_{\omega_{\R}}(E), A)
+ 
2 \sqrt{2} (\omega_{\R}(N_{\R}^2) - \omega_{\R}(N_{\R})^2)^{1/2}
D(\Gamma_{\omega_{\R}}(E), A)^{1/2},   
\end{eqnarray*}
which shows a trade-off relation between 
the accuracy and the size of the reference frame.} 
\end{example}

%%%%%%

%\subsection{Discussion}

\noindent{\sl Discussion.} The necessity of the large reference frame, in the sense of large fluctuation $\omega_{\R}(N_{\R}^2 ) -\omega_{\R}(N_{\R})^2$, can be 
interpreted in terms of the uncertainty relation for joint measurability 
\cite{miuncertainty}. 
$E$ commutes with $N:=N_{\Sy} + N_{\R}$, and one may consider
joint measurements of $E$ and $N$ where these observables may be regarded as approximators 
to $A$ and $N_\Sy$, respectively.
The aim is to obtain outcome distributions of $E$ 
that are close to those of $A$. Since $A$ and $N_{\Sy}$ do not commute, 
this closeness entails, due to the joint measurement uncertainty relation, that the outcome distribution of $N$ should contain less
``information" on $N_{\Sy}$. In other words, the (hypothetical) 
measurement of $N$ must be a highly inaccurate measurement of 
$N_\Sy$. This large approximation error can be achieved only by the 
initial uncertainty of $N_{\R}$. Put another way, for $A$ to be a good 
``absolute" representative of $E$ in the context of a joint measurement of $E$ and $N$, by the uncertainty relation for $A$ and $N_{\Sy}$, $N_{\Sy}$ cannot well approximate $N$,
and this discrepancy between $N_{\Sy}$ and $N$ is afforded by a large spread in $N_{\R}$.

We also note that symmetry constraints of the form we have considered are relevant in the framework of resource theories (e.g., \cite{ms1}, \cite{ajr1}), wherein a reference state $\omega_{\R}$ is viewed as a resource if it is not invariant under $U(1)$  symmetry transformations. The asymmetric state is seen in \cite{ms1} as allowing for the simulation of non-invariant statistics of $\Sy$ under the constraint of symmetric operations of $\Sy + \R$, and the degree of asymmetry may be quantified and exploited in the Wigner-Araki-Yanase theorem \cite{ajr1}. The connection between asymmetry and localisation, as we have it, is clear: highly localised reference states, with respect to a phase conjugate to number, are highly asymmetric under phase shifts. The exact equivalence, should there be one, between localisation as allowing non-invariant quantities to represent relative ones on the one hand, and asymmetry as a resource for good measurements in Wigner-Araki-Yanase-type scenarios on the other, remains a task to be investigated. 

\section{Conclusion}\label{sec:con}
Through model considerations and generic trade-off relations we have shown
that the possibility of the traditional tacit reduction of relative quantities to ``absolute" quantities
is contingent upon the size of the reference system and a judicious choice of reference 
(system) state. An obstruction to complete specification of subsystem statistics, or rather to subsystem quantities
being used to approximate the corresponding invariant quantities, arises
due to the incompatibility of subsystem quantities. 

The necessity of a large uncertainty in the symmetry generator in a given reference state for good approximation is an essentially quantum feature. In classical mechanics, all observables commute and all (pure) states are well localised
with respect to all classical quantities. Provided that the classical reference system admits a faithful action of the symmetry group, this reference system is sufficient for the ``absolute"
quantities to perfectly represent the relative ones, with no constraint on the values of other quantities at all.

The analysis presented here constitutes a first step towards a comprehensive, operational understanding of the role of (quantum) reference systems in the description of quantum experiments, and suggests a number of further avenues of exploration, for example the physically relevant case of non-Abelian symmetries, approximation of invariant observables by ``absolute" phase-space quantities, strength of Bell inequality violation for finite-size reference systems, and many more.

\bigskip
\noindent{\bf Acknowledgements} Many thanks to Tom Bullock for a careful reading of an earlier version of this manuscript.
TM acknowledges JSPS KAKENHI (grant no. 15K04998). LL acknowledges support under the grant \emph{Quantum Mathematics and Computation} (no.  EP/K015478/1).


\section*{References}

\begin{thebibliography}{99}

\bibitem{Wigner} E. P.~Wigner, 1967 {\it Symmetries and Reflections} (Bloomington and London: Indiana University Press)

\bibitem{Weyl} H.~Weyl 1931. The Theory of Groups and Quantum Mechanics,  Reprint, Dover, 1950.

\bibitem{brs}S.~D.~Bartlett, T.~Rudolph and R.~W.~Spekkens, 2007. 
Reference frames, superselection rules, and quantum information, 
Rev. Mod. Phys. {\bf 79}, 555. 


\bibitem{lbm1} L.~Loveridge, P.~Busch and T.~Miyadera, 2016. Quantum measurements and symmetry, \emph{in preparation}.

\bibitem{Page}D.~N.~Page, W.~K.~Wootters, 1983.
Evolution without evolution: Dynamics described by stationary observables, 
Phys.~Rev.~D {\bf 27}, 2885. 

\bibitem{Milburn}G.~J.~Milburn, D.~Poulin, 2006. 
Relational time for systems of oscillators, Int. J. Quantum Inform., 
{\bf 04}, 151.  

\bibitem{Brunetti}R.~Brunetti, K.~Fredenhagen, M.~Hoge, 2010. 
Time in quantum physics: from an external parameter to an 
intrinsic observable, 
Found.~Phys. {\bf 40}, 1368.


\bibitem{Giovannetti}V.~Giovannetti, S.~Lloyd, L.~Maccone, 
2015. Quantum time, Phys.~Rev.~D {\bf 92}, 045033. 


\bibitem{wig1} E.~Wigner, 1952. Die Messung quantenmechanischer Operatoren, Z. Phys. {\bf 133}, 101.

\bibitem{buschtrans} E. P.~Wigner, 1952. Measurement of quantum mechanical operators.
Translation of [5] by P.~Busch; http://arxiv.org/abs/1012.4372.


\bibitem{ay1} H.~Araki and M.M.~Yanase, 1960. Measurement of quantum mechanical operators, Phys. Rev. {\bf 120}, 622.

\bibitem{yan1} M. M.~Yanase, 1961. Optimal measuring apparatus, Phys. Rev. 
{\bf 123}, 666.

\bibitem{mi1} T.~Miyadera and H.~Imai, 2006. 
Wigner-Araki-Yanase theorem on distinguishability, 
Phys. Rev. A {\bf 74}, 024101.

\bibitem{lb1} L.~Loveridge and P.~Busch 2011. ``Measurement of quantum mechanical operators" revisited, Eur. Phys. J. D. {\bf 62}, 2.

\bibitem{lb2} P.~Busch and L.~Loveridge 2011. Position measurements obeying momentum conservation, Phys. Rev. Lett. {\bf 106}, 110406.

\bibitem{ms1} I.~Marvian and R. W.~Spekkens, 2012. An information-theoretic account of the Wigner-Araki-Yanase theorem, arXiv:1212.3378. 


\bibitem{www} G. C.~Wick, A. S.~Wightman, and E. P.~Wigner, 1952. Intrinsic parity of elementary particles, Phys. Rev. {\bf 88}, 101.

\bibitem{as} Y.~Aharonov and L.~Susskind, 1967. Charge superselection rule, Phys. Rev. {\bf 155}, 1428.

\bibitem{buschunsharp}
P.~Busch, 2009. On the sharpness and bias of quantum effects, 
Found.~Phys. {\bf 39}, 712-730.

\bibitem{miuncertainty}
T.~Miyadera and H.~Imai, 2008. Heisenberg's uncertainty principle for simultaneous measurement of positive-operator-valued measures, 
Phys.~Rev.~A {\bf 78}, 052119. 

\bibitem{polterovich}
L.~Polterovich, 2014. Symplectic geometry of quantum noise, 
Commun.~Math.~Phys. {\bf 327}, 481-519. 

\bibitem{Davies}
E.~B.~Davies, 1976. {\em Quantum Theory of Open Systems}, Academic Press.  


\bibitem{QMMT}
P.~Busch, P.~Lahti, J.-P.~Pellonp\"a\"a and K.~Ylinen, {\em Quantum Measurement},
Springer (forthcoming).

\bibitem{numphasecomplem}
P.~Busch, P.~Lahti, J.-P.~Pellonp\"a\"a and K.~Ylinen, 2001. Are number and phase complementary observables?,
J.~Phys.~A {\bf 34}, 5923-5935,

\bibitem{Fuchs}
C.~A.~Fuchs, C.~M.~Caves, 1994. 
Ensemble-dependent bounds for accessible information in quantum mechanics, 
Phys.~Rev.~Lett. {\bf 73}, 3047.

\bibitem{Nielsen}
 M.~Nielsen, I.~Chuang, 2000. 
{\em Quantum Computation and Quantum Information}, Cambridge University Press.

\bibitem{Jozsa}H.~Barnum, C.~M.~Caves, 
C.~A.~Fuchs, R.~Jozsa, B.~Schumacher, 
1996. 
Noncommuting mixed states cannot be broadcast,
Phys.~Rev.~Lett. {\bf 76}, 2818. 

\bibitem{Mandelstam}
L.~I.~Mandelshtam, I.~E.~Tamm, 
1945. The uncertainty relation between energy and time 
 in nonrelativistic quantum mechanics, 
J.~Phys.~U.S.S.R. {\bf 9}, 249. 

\bibitem{BuschTime}
P.~Busch, 2008. 
The time-energy uncertainty relation, 
in {\em Time in Quantum Physics}, Springer. 

\bibitem{MiTime}
T.~Miyadera, 2015. 
Energy-time uncertainty relations in quantum measurements, 
arXiv:1505.03707.

\bibitem{BuschOperational}
P.~Busch, M.~Grabowki, P.~Lahti, 1997. {\em Operational Quantum Physics}, 
Springer, 2nd ed.

\bibitem{miJMP}
T.~Miyadera, 2011.
Uncertainty relations for joint localizability and joint measurability in 
finite-dimensional systems, 
J.~Math.~Phys. {\bf 52}, 072105. 

\bibitem{Ja06} B.~Janssens, 2006. Unifying decoherence and the Heisenberg Principle, arXiv:quant-ph/0606093.

\bibitem{ajr1} M.~Ahmadi, D.~Jennings and T.~Rudolph, 2013. The Wigner-Araki-Yanase theorem and the quantum resource theory of asymmetry, New J. Phys. {\bf 15}, 013057.



\end{thebibliography}
\end{document}